\documentclass[a4paper,preprintnumbers,floatfix,superscriptaddress,pra,twocolumn,showpacs,notitlepage,longbibliography]{revtex4-1}
\usepackage[utf8]{inputenc}
\usepackage[T1]{fontenc}
\usepackage[sc,osf]{mathpazo}
\usepackage{amsmath}
\usepackage{amsthm}
\usepackage{bbm}
\usepackage{comment} 
\usepackage{marvosym}
\usepackage{graphicx}
\usepackage{xcolor}
\usepackage{appendix}

\usepackage{float}
\usepackage{xr-hyper} 
\usepackage{xcite} 
\usepackage{hyperref}
\usepackage{physics}
\usepackage{braket}

\usepackage[normalem]{ulem}
\makeatletter

\newcommand{\beq}{\begin{eqnarray}}
\newcommand{\eeq}{\end{eqnarray}}

\newcommand{\mean}[1]{\ensuremath{\langle{#1}\rangle}}

\newcommand{\bounded}[2]{
\mathcal{B}(\mathbb{C}^{#1}\otimes\mathbb{C}^{#2})}

\newtheorem{observation}{Observation}
\newtheorem{theorem}[observation]{Theorem}

\def\1{\mathbf{1}}
\def\0{\mathbf{0}}
\def\H{\mathcal{H}}

\def\fsub{\subset_{\mathrm{L}}}
\def\fsup{\supset_{\mathrm{L}}}

\def\Tb{{T_{\mathrm{B}}}}
\def\Ta{{T_{\mathrm{A}}}}

\def\lB{_{\mathrm{B}}}
\def\lA{_{\mathrm{A}}}
\newtheorem{observation }{Observation}

\usepackage{lineno}
\usepackage{balance}

\begin{document}

\title{Scaling Bound Entanglement through Local Extensions}

\author{Robin Krebs}
  \email{robinbenedikt.krebs@stud.tu-darmstadt.de}
 \affiliation{
 Department of Computer Science, Technical University of Darmstadt, Germany
}
\author{Mariami Gachechiladze}%
 \email{mariami.gachechiladze@tu-darmstadt.de}
\affiliation{
 Department of Computer Science, Technical University of Darmstadt, Germany
}

\begin{abstract}
Entanglement is a central resource in quantum information science, yet its structure in high dimensions remains notoriously difficult to characterize. One of the few general results on high-dimensional entanglement is given by peel-off theorems, which relate the entanglement of a state to that of its lower-dimensional local projections. We build on this idea by introducing \emph{local extensions}, the inverse process to peel-off projections, which provide a systematic way to construct higher-dimensional entangled states from lower-dimensional ones. This dual perspective leads to general bounds on how the Schmidt number can change under projections and extensions, and reveals new mechanisms for generating bound entangled states of higher dimensionality. As a concrete application, we construct a positive-partial-transpose state of Schmidt number three in local dimensions $4\times 5$, the smallest system known to host such entanglement. We further extend this approach to identify an elegant family of \textit{generalized grid states} with increasing Schmidt number, including explicit examples of a $7\times 7$ state with Schmidt number four and a $9\times 9$ state with Schmidt number five, suggesting $(d+1)/2$ scaling in odd local dimensions $d\times d$. Taken together, our results provide a constructive toolkit for probing the scaling of bound entanglement in high dimensions.
\end{abstract}
\maketitle
\textit{Introduction -- } Quantum entanglement is a defining feature of quantum mechanics and a central resource 
for quantum technologies, underpinning applications such as quantum key distribution~\cite{ekert1991quantum,gisin2002quantum},  networks~\cite{kimble2008quantum,wehner2018quantum},  computation~\cite{raussendorf2001one},  metrology~\cite{giovannetti2004quantum}, etc. While a sophisticated mathematical theory of entanglement 
has been developed~\cite{horodecki2009quantum,plenio2005introduction,guhne2009entanglement}, many fundamental questions remain open. In particular, deciding 
whether a state is separable is NP-hard~\cite{gurvits2004classical,gurvits2003classical,gharibian2008strong}, with efficient criteria available only in 
low-dimensional cases, such as the qubit-qubit and qubit-qutrit systems,~\cite{peres1996separability, horodecki1997separability} owing to a well-known structure of positive maps in these dimensions~\cite{woronowicz1976positive,stormer2013decomposition}. For higher dimensions, no general 
method for entanglement detection is known.  At the same time, experimental progress has made high-dimensional entanglement accessible, with advantages such as improved noise resistance and enhanced performance in communication and computation tasks~\cite{kaszlikowski2000violations,howell2002experimental,mair2001entanglement,krenn2017orbital,erhard2020advances,hu2020experimental}. A central challenge is thus to certify whether 
a given experiment produces genuinely high-dimensional entanglement, or whether its 
results admit a lower-dimensional explanation~\cite{kraft2018characterizing}. A natural measure for this purpose is 
the Schmidt number, which quantifies the number of degrees of freedom across which 
a state is entangled~\cite{terhal2000schmidt,guhne2002detection,lewenstein2000optimization}. However, estimating the Schmidt number is notoriously difficult, 
and only a few methods with limited applicability are known.  

A particularly intriguing class of states are those with positive partial transpose (PPT)~\cite{peres1996separability}. These bound entangled 
states cannot be distilled into pure entanglement and were long regarded as of limited 
use~\cite{horodecki1998mixed}. This view has changed: PPT states have been shown to provide advantages in tasks 
such as steering~\cite{moroder2014steering}, nonlocality~\cite{vertesi2014disproving}, and secure communication~\cite{horodecki2005secure}, and families with increasing 
Schmidt number have been constructed~\cite{bennett1999unextendible,chen2006quantum,huber2018high,pal2019class,krebs2024high}. Yet the existence and structure of high Schmidt number PPT 
states remain poorly understood, and systematic methods for their construction are 
largely absent. Progress in this direction would not only improve our ability to certify entanglement in 
practical settings, but also shed light on the fundamental interplay between 
entanglement dimensionality and positivity. In particular, the existence of high 
Schmidt number PPT states is tied to open problems on indecomposable $k$-positive maps 
\cite{woronowicz1976positive,skowronek2009cones,mlynik2025characterization}, and relates 
to broader questions such as the PPT$^2$-conjecture \cite{chen2019positive} and the 
possibility of non-PPT bound entanglement \cite{muller2016positivity}.

In this work, we introduce the method of \emph{local extensions}, the natural dual of 
peel-off projections, which provides a constructive framework for generating 
higher-dimensional entangled states while preserving the PPT property. Peel-Off 
theorems and related projection properties have played an important role in the 
literature~\cite{yang2016all} by constraining how much entanglement can survive under local dimension 
reduction, thereby giving one of the few general insights into high-dimensional 
entanglement. Our approach complements this line of work by showing how the 
converse operation, extending local dimensions, can be used to systematically 
construct new families of PPT states with large Schmidt number. This perspective 
yields general bounds on how the Schmidt number transforms under projections and 
extensions, establishing a structural principle for controlled increases in 
entanglement dimensionality.

Building on this framework, we construct a PPT state of 
Schmidt number three in local dimensions $4 \times 5$, the smallest system known to 
host such entanglement, and develop generalized grid states whose iterative extensions 
yield explicit examples with Schmidt numbers four and five in dimensions $7 \times 7$ and 
$9 \times 9$, respectively. Together, these results improve the best known 
dimensional thresholds and suggest an underlying $(d+1)/2$ scaling of Schmidt number in 
odd $d \times d$ systems. Altogether, our findings position local extensions as a 
systematic toolkit for constructing and certifying high-dimensional bound entangled 
states, and advance the broader effort to chart the structure of entanglement in 
mixed-state quantum systems, especially in the context of distillability and incompressibility of quantum resources~\cite{weinbrenner2024superactivation}.

\textit{Preliminaries --}  A quantum state is a positive semidefinite operator $\rho\geq 0$ in a finite dimensional Hilbert space. 
Bipartite quantum states  $\rho\in\mathcal{B}(\mathcal{H}_\mathrm{A}\otimes \mathcal{H}_\mathrm{B})$ with $\rho^{\Tb} := (\mathrm{id}_{\mathrm{A}}\otimes T_B)(\rho_\mathrm{AB})\ge 0$ are termed positive partial transpose, or shortly, PPT states~\cite{peres1996separability}. All states that do not have this property are entangled~\cite{choi1972positive}. 
PPT states are a subset of bound entangled states from which no Bell pair can be distilled using arbitrary operations of the form $(A\otimes B)\rho(A\otimes B)^\dagger$, which are called
SLOCC-operations~\cite{vidal1999entanglement}. 
We refer to SLOCC operations with isometric $A^\dagger\in \mathcal{B}(\mathbbm{C}^m,\mathcal{H}_\mathrm{A})$, $B^\dagger\in \mathcal{B}(\mathbbm{C}^n,\mathcal{H}_\mathrm{B})$ as projections onto a local block, i.e. a subspace of the form $\mathcal{U}\otimes\mathcal{V}= (A^\dagger \mathbbm{C}^m)\otimes(B^\dagger \mathbbm{C}^n)\subseteq \mathcal{H}_\mathrm{A}\otimes \mathcal{H}_\mathrm{B}$. 
The PPT set is further closed under arbitrary tensor products and SLOCC operations~\cite{clarisse2005characterization}.

For a bipartite  state $\rho$, the Schmidt number measures the dimensionality of entanglement, defined as
\begin{align}
\begin{split}
    \mathrm{SN}(\rho):= & \min~\  k \\
    \mathrm{s.t. }~ \rho =& \sum_i p_i\ketbra{\psi_i}{\psi_i},\  \mathrm{SR}(\ket{\psi_i})\le k,\ ~p_i>0,
\end{split}
\end{align}
where Schmidt Rank (SR)
corresponds to the number of nonzero Schmidt coefficients. $\mathrm{SN}(\rho)=1$ iff $\rho$ is separable, which implies the  complexity of computing Schmidt numbers is at least as hard as deciding separability. 
A map $\phi$ is termed $k$-positive, if for every $\rho$ with $\mathrm{SN}(\rho)\le k$, the following holds: $(\phi_\mathrm{A}\otimes \mathrm{id}_\mathrm{B})(\rho)\ge 0$.

\textit{Local extension method for high Schmidt number states--} High Schmidt number PPT states known in literature~\cite{pal2019class,krebs2024high}, have an interesting nested structure: they contain local PPT blocks of decreasing Schmidt number. In this section, we introduce a reverse concept, which we call \textit{local extension} $\rho$ of $\rho_c$, and denote it by $\rho_c\fsub \rho$. We say  $\rho_c\fsub \rho$, if $\rho_c$ can be obtained by projecting to a local block of $\rho$ using SLOCC operators. The local extension  construction will help to find high Schmidt number PPT states.

\begin{observation}~\label{obs:matrioshka}
    Every Schmidt number $k$ state $\rho^{(k)}$ is the local extension of a Schmidt number $(k-1)$ state, that is
    \begin{equation}
        \rho^{(1)}\fsub \dots \fsub \rho^{(k-1)}\fsub \rho^{(k)}
    \end{equation}
    exists, satisfying $\mathrm{SN}(\rho^{(i)})=i$.
\end{observation}
\begin{proof}
This result relies on the Choi decomposition theorem that we state here for completeness~\cite{yang2016all}. Every $k$-positive map $\phi$ from $\mathcal{B}(\mathbb{C}^{m})$ to $\mathcal{B}(\mathbb{C}^{n})$, $m,n\ge 1$, $k\geq 2$ can be decomposed as $\psi + \gamma\circ \eta$, where $\psi$ is completely positive, $\gamma$ is a $(k-1)$-positive map from $\mathcal{B}(\mathbb{C}^{m-1})$ to $\mathcal{B}(\mathbb{C}^{n})$ and $\eta(\rho) = V\rho V^\dagger$, where $V^\dagger$ is an isometry from $\mathbb{C}^{m-1}$ into $\mathbb{C}^{m}$.

Suppose $(\phi\otimes\mathrm{id}) (\rho)\not\ge 0$ for $\mathrm{SN}(\rho)=k+1$. Since $(\psi\otimes\mathrm{id})(\rho) \ge 0$,  it follows that $(\gamma\otimes\mathrm{id})(\eta\otimes\mathrm{id})(\rho)\not\ge 0$. \\
As $(\eta\otimes\mathrm{id})(\rho)$ is a projection of $\rho\in\bounded{m}{n}$ to $\mathcal{B}(\mathbb{C}^{m-1}\otimes\mathbb{C}^{n})$,  this projection $(\eta\otimes\mathrm{id})(\rho)$ must have a Schmidt number of at least $k$ for $(\gamma\circ\eta \otimes\mathrm{id})(\rho)\not\ge 0$ to hold. 
Iterating the same argument with the map and state  $\gamma$, $(\eta\otimes\mathrm{id})(\rho)$ in place of $\phi$, $\rho$, induces the subsequent nested steps.
\end{proof}

 For a given unnormalized state $\rho_c$ in $m\times n$, which we call a \textit{core} state, we say that $\rho$ is its $(\delta_{\mathrm{A}},\delta_{\mathrm{B}})$-extension, if  $\rho_c \fsub \rho $ in $(m+\delta_{\mathrm{A}})\times (n+\delta_{\mathrm{B}})$, with $\delta_{\mathrm{A}},\delta_{\mathrm{B}}\in \mathbbm{N}$.
 In what follows, we study the $(1,0)$-extensions, if not further specified, i.e. subsystem $\mathrm{A}$ is extended by adding a state $\ket{\bot}$ to $\mathbbm{C}^m$.
This adds a new diagonal block supported on $\ket{\bot}\otimes\mathbbm{C}^n$, referred to as $\rho_e$, and new off-diagonal elements $\chi \in \mathcal{B}(\mathbbm{C}^n, \mathbbm{C}^m\otimes \mathbbm{C}^n)$. Hence, the full local extension may be written as unnormalized state,
\begin{equation}~\label{eq:block_form}
   \hspace{-0.15cm} \rho = \rho_c + \chi\bra{\bot}+ \ket{\bot}\chi^\dagger + \ket{\bot}\bra{\bot}\otimes\rho_e = \begin{pmatrix}
        \rho_c & \chi\\
        \chi^\dagger & \rho_e        
    \end{pmatrix}.
\end{equation}

Since we require that local extensions $\rho\geq 0$, we learn the relations between ranges, $R(\chi) \subseteq R(\rho_c)$, and $R(\chi^\dagger)\subseteq R(\rho_e)$. Furthermore, by the (generalized) Schur complement $\rho_{e\backslash c} \equiv \rho_e-\chi^\dagger\rho_c^{-1}\chi\ge 0$ and $\rho_{c\backslash e} \equiv \rho_c-\chi\rho_e^{-1}\chi^\dagger\ge 0$. See SM~\ref{app:ext_psd} and~\ref{app:ext_ppt} for a more detailed discussion of the convex geometry of local extensions.  
We are ready to formulate a strengthened version of Observation~\ref{obs:matrioshka}:

\begin{theorem}~\label{th:ext2}
    Let $\rho\in\mathcal{B}(\mathcal{H}_\mathrm{A}\otimes \mathcal{H}_\mathrm{B})$ and $\ket{\phi}\in \H_\mathrm{A}$.  Then, the Schmidt number $\mathrm{SN}{((A\otimes\mathbb{1}) \rho(A^\dagger\otimes\mathbb{1}))}$ drops at most by $1$, compared to $\mathrm{SN}(\rho)$, if $A=\mathbb{1}-\ketbra{\phi}{\phi}$.
\end{theorem}
The above statement can be reversed to refer to local extensions, yielding: A $(1,0)$-extension cannot increase the Schmidt number by more than one. This leads to an immediate observation that

\begin{observation}
    Any $3\times 4$ PPT state $\tilde\rho$ with a product vector $\ket{\alpha\beta}$ in its kernel has Schmidt number at most $2$.   
\end{observation}
To see this, project $\tilde\rho$ on a $2\times 4$ block which contains the product state $\ket{\alpha\beta}$. Then the resulting state is PPT with the product vector in its kernel. Such states are known to be separable by Ref.~\cite{kraus2000separability}. Theorem~\ref{th:ext2} then implies that $\tilde\rho$ has Schmidt number at most $2$.
We now prove Theorem~\ref{th:ext2} (See Supplemental material (SM)~\ref{app:SCpeel} for the equivalent statement in terms of $k$-positive maps).

\begin{proof}
We use the notation in Eq.~(\ref{eq:block_form}). We show that for any conic decomposition of $\rho_c$ into vectors  $\{\ket{\psi^{c}_i}\}_{i<N}$ we can induce a conic decomposition $\rho = \sum_i \ket{\tilde{\psi}_i}\bra{\tilde{\psi}_i}+\rho^{\mathrm{sep}}$, with $\mathrm{SR}(\ket{\tilde{\psi}_i})\le \mathrm{SR}(\ket{{\psi}^{c}_i})+1$ and a separable $\rho^{\mathrm{sep}}$,
with the following \textit{ansatz}:
\begin{equation}
    \ket{\tilde{\psi}_i} = \rho \begin{pmatrix}
        \rho_c^{-1}\ket{\psi^{c}_i}\\ 0\end{pmatrix} =
        \begin{pmatrix}
        \,\ket{\psi^{c}_i}\\  
        \chi^\dagger\rho_c^{-1}\ket{\psi^{c}_i}
    \end{pmatrix} \in R(\rho), 
\end{equation}
\begin{align}
    \tilde{\rho} \equiv \sum_i \ket{\tilde{\psi}_i}\bra{\tilde{\psi}_i} = 
    \begin{pmatrix}
        \rho_c & \chi \\
        \chi^\dagger & \chi^\dagger \rho_c^{-1}\chi
    \end{pmatrix} ,
\end{align}
where $\rho_c^{-1}$ is the pseudoinverse and we used identity $\rho_c\rho_c^{-1}\ket{\psi^{c}_i} = \ket{\psi^{c}_i}$ and the identity $\rho_c\rho_c^{-1}\chi=\chi$. 
Using the Schur complement, we see that $\rho-\tilde{\rho} =  \rho_{e\backslash c}\otimes\ketbra{\bot}{\bot}\ge 0$, is 
 trivially separable and, thus, can be added to $\tilde{\rho}$ without increasing its Schmidt number.  Also, 
$\mathrm{SR}(\ket{\tilde{\psi}_i} )\le  1+\mathrm{SR}(\ket{\psi^{c}_i})$
, as $\mathrm{SR}(\ket{\tilde{\psi}_i}-\ket{{\psi}^{c}_i} )\le 1$. Now suppose that $\mathrm{SR}(\ket{\psi^{c}_i})\le k$, and that this is the minimum $k$ for which a decomposition exists. Then, Theorem~\ref{th:ext2} follows automatically. \end{proof}
%--------------------------------------------

\textit{Results on PPT local extensions} --
Every PPT state permits PPT local extensions, but we are interested in those that may increase the Schmidt number. Local extensions obtained by SLOCC operations, given by $S_\mathrm{A}\in\mathcal{B}(\mathcal{H}_\mathrm{A},\mathcal{H}_\mathrm{A}\oplus\ket{\bot})$, where $\ket{\phi}\in \H_{\mathrm{A}}$, 
\begin{equation}\label{eq:SLOCC_ext}
    S_\mathrm{A}=\mathbb{1}(\H_\mathrm{A})+ \ket{\bot}\bra{\phi},
\end{equation}
define a set of trivial local extensions regarding the Schmidt number.
We call the PPT states that only allow SLOCC local extensions \textit{unextendible}. 

Here we develop a more general local extension technique for PPT states. 
If subsystem $\mathcal{H}_\mathrm{A}$ is locally extended (provided we define $\ket{\bot} = \ket{\bot}^*$, without loss of generality), the partial transposition takes the form 
\begin{equation}
 \rho^{\Ta} = 
\begin{pmatrix}
        \rho_c ^{\Ta}& (\chi^{\dagger})^{\Ta}\\
        \chi^{\Ta} & \rho_e        
\end{pmatrix}.
\end{equation}
Then, $\left(\rho^{\Ta}\right)_{e\backslash c}\equiv \rho_e -\chi^{\Ta}(\rho_c ^{\Ta})^{-1} (\chi^{\dagger})^{\Ta}\geq0$.
The main idea is to identify $\chi$, which leads to a valid PPT local extension by setting up a system of linear equations. First, observe that by Schur complement for a fixed $\rho_c$, to ensure PPT, the identities $P\chi=\chi$ and $Q(\chi^\dagger)^{\Ta}=(\chi^\dagger)^{\Ta}$ need to hold, where $P = \mathbb{1}(R(\rho))$, $Q = \mathbb{1}(R(\rho^{\Ta}))$ are projectors on the ranges, respectively. 
We rewrite these constraints in $\chi$ in the tripartite Choi dual form, 
\begin{equation}
\label{eq:linext}
    P_{\mathrm{AB}}\otimes\mathbb{1}_{\bar{\mathrm{B}}}\ket{\tilde\chi}_{\mathrm{A}\mathrm{B}\bar{\mathrm{B}}}=    Q_{\mathrm{A}\bar{\mathrm{B}}}\otimes\mathbb{1}_\mathrm{B}\ket{\tilde\chi}_{\mathrm{A}\mathrm{B}\bar{\mathrm{B}}}=
    \ket{\tilde\chi}_{\mathrm{A}\mathrm{B}\bar{\mathrm{B}}},
\end{equation}
where
$\ket{\tilde\chi}_{\mathrm{A}\mathrm{B}\bar{\mathrm{B}}} = \chi\ket{\Gamma}_{\mathrm{B}\bar{\mathrm{B}}}$, $\ket{\Gamma} = \sum_i \ket{ii} \in \H_{\mathrm{B}}\otimes\H_{\bar{\mathrm{B}}}$, and 
the partial transposition now corresponds to the SWAP operations between subsystems B and $\bar{\mathrm{B}}$.

These linear equations allow computing a lower bound on the dimension of the solution space. From a full space of dimension $mn^2$ we subtract the dimension of the kernels of $(P_{\mathrm{AB}}\otimes \mathbb{1}_{\bar{\mathrm{B}}})$ and $(Q_{\mathrm{A}\bar{\mathrm{B}}}\otimes \mathbb{1}_\mathrm{B})$, which are $(mn-p)n$, and $(mn-q)n$, where $p,q$ are the ranks of $\rho_c$ and $\rho_c^{T_A}$, respectively, forming a \textit{birank} $(p,q)$.  The number of nontrivial extensions is lower-bounded by
\begin{equation}
    \#(\mathrm{local extensions}) \ge (p+q-mn)n \label{eq:no_extensions}-m.
\end{equation}

For generic states where the right-hand side of the inequality is negative, numerically, we only find unextendible states. To check this, we generate random states with a fixed birank using the method from Ref.~\cite{leinaas2007extreme}. Specifically, to create a PPT state of birank $(p,q)$, we take the $(nm - p)$ smallest eigenvalues of a random Hermitian matrix $X$, with normally distributed entries, and the $(nm - q)$ smallest eigenvalues of its partial transpose $X^{T_B}$, combine them into a single vector, and then drive this vector to zero using the Gauss-Newton method. 
Sampling $100$ states for system sizes $2\times 4$, $3\times 3$, $3\times 4$, and $4\times 4$ and all biranks shows that the inequality is empirically tight for small systems, when the obtained state is an extreme point of the PPT set. Next, we use a parameter-counting argument to derive analytical results for the allowed PPT entanglement structures in the low-dimensional cases. 
We show that

\begin{theorem}
    Every $3\times 3$ PPT entangled state with rank $4$ is unextendible.
\end{theorem}
\begin{proof}
It is known that every $3\times 3$ PPT entangled state $\rho$ of rank $4$  has five linearly independent and exactly one linearly dependent product states $\ket{\alpha_i \beta_i} \in \mathbbm{C}^3\otimes \mathbbm{C}^3$,  $0< i \leq 6$ in its kernel~\cite{chen2011description}. Using Eq.~\eqref{eq:linext}, for $0<i<6$, any local extension needs to satisfy the set of linear equations, 
$\bra{\alpha_i\beta_i}_{{\mathrm{AB}}}\otimes \mathbb{1}_{\bar{\mathrm{B}}}\ket{\chi}_{\mathrm{AB}\bar{\mathrm{B}}}=0 $ and $\bra{\alpha_i}_{{\mathrm{A}}}\otimes \mathbb{1}_{{\mathrm{B}}}\otimes (\bra{{\beta}_i})^*_{\bar{\mathrm{B}}}\ket{\chi}_{\mathrm{AB}\bar{\mathrm{B}}}=0$.
This gives us in total $30$ linear equations in $\ket{\chi}_{\mathrm{AB}\bar{\mathrm{B}}}$, but exactly $6$ of them are redundant, as $\braket{\alpha_i\beta_i{(\beta_i)}^*|\chi}=0,~(0<i\le 6)$.
The $d_\mathrm{A}(d_\mathrm{B})^2=27$ dimensional solution space reduces to $3$ dimensions after imposing these $24$ linearly independent constraints. 
The SLOCC extensions fully realize a $3$-dimensional solution space, and thus only trivial extensions may exist. 
\end{proof}

Similar results can be derived for the cases where the exact characterization of the product state in the kernel is known. For example, all $4\times 4$ PPT states with rank~$6$~\cite{leinaas2010numerical} are also unextendible. 
Next, we identify a sufficient condition for the existence of nontrivial PPT local extensions for $\rho_c$ that is independent of the birank inequalities. 
Contrary to the previous unextendible examples, where the kernel is spanned by product states and the range is a completely entangled subspace, we see that states with product vectors in the range are easy to extend nontrivially.

\textit{Sufficient conditions for nontrivial PPT extensions --}
Suppose there exist a product vector  in the range of the core state $\ket{\alpha\beta}\in R(\rho_c)$ and one in the corange $\ket{\alpha\gamma}\in R(\rho^{T_\mathrm{A}}_c)$ with $|\mean{\beta|\gamma}| \neq 1$ and $\mathrm{rk}(\bra{\alpha} \rho_c \ket{\alpha})>2$. Then, there exists a PPT local extension with $\chi=\ketbra{\alpha\beta}{\gamma}$, which satisfies Eq.~\eqref{eq:linext}. 
Curiously, in this construction, the associated state $\ket{\Tilde{\chi}}_{\mathrm{AB}\bar{B}}$ is a product state, and thus contains no entanglement. However, this simple form of an extension can be used to increase the Schmidt number.
Finally, we choose $\rho_e$ with a minimal rank,  $\rho_e= \bra{\alpha\beta}\rho_c^{-1}\ket{\alpha\beta}\ketbra{\gamma}{\gamma}+\bra{\alpha\gamma}(\rho^{\Ta})^{-1}\ket{\alpha\gamma}\ketbra{\beta}{\beta}$.
This local extension is always nontrivial. In what follows, we show that it can increase the Schmidt number of a PPT state, and with it, we discover the highest Schmidt number PPT states known.  

\begin{figure}
    \includegraphics[width=0.99\linewidth]{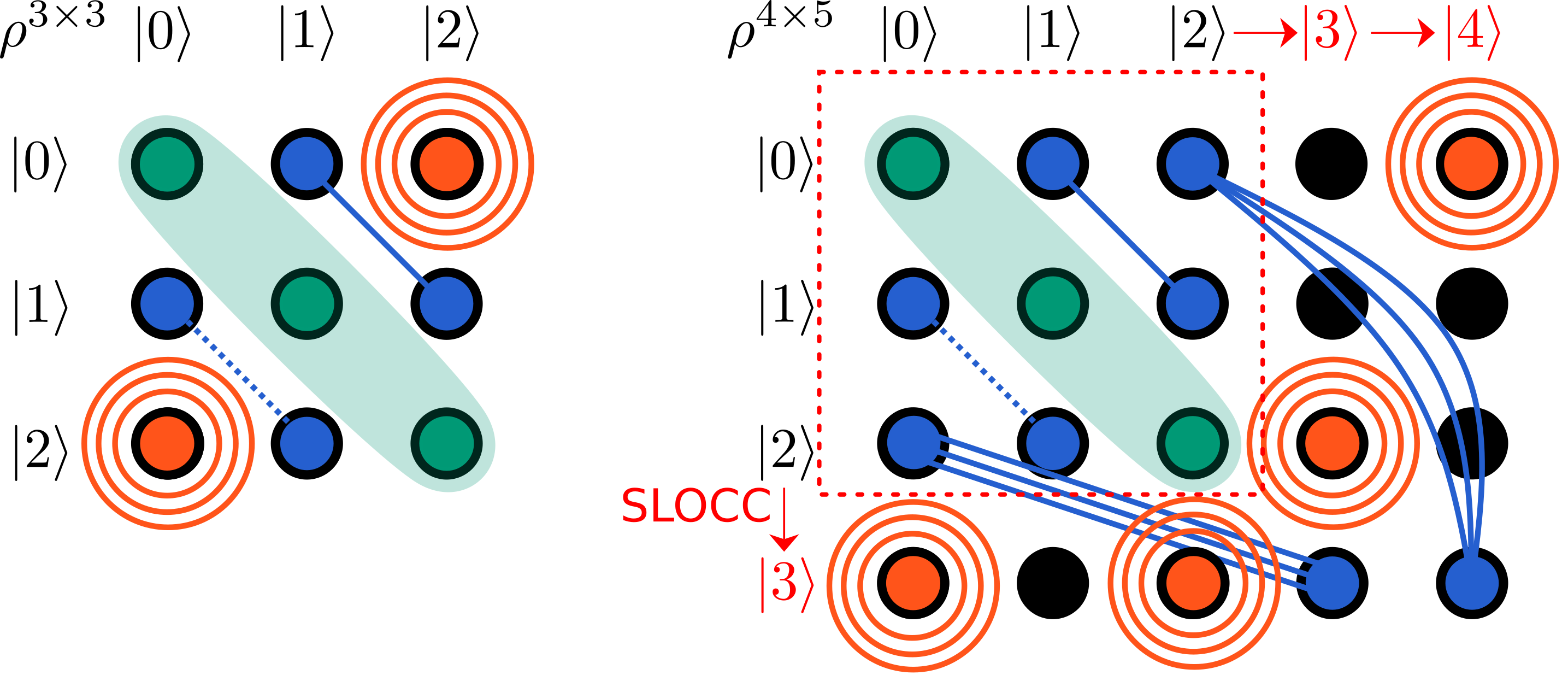}
    \caption{Generalized grid states: (left) Schmidt number $2$ PPT bound entangled state in $3\times 3$ local dimensions, (right) Schmidt number $3$ PPT bound entangled state in $4\times 5$ obtained after local extensions from the $\rho^{3\times 3}$. First SLOCC extension is done by subsystem $A$, then a sequence of two nontrivial extensions in performed by subsystem $B$. The red dashed rectangle indicates the original $3\times 3$ block. } \label{fig:45be}
\end{figure}

\begin{theorem}~\label{th:4x5ex}
    In local dimensions $4\times 5$ a Schmidt number $3$ PPT state exists.
\end{theorem}

We obtain this state by $(1,2)$-local extension of a PPT $\rho^{3\times 3}$  in $3\times 3$, which has Schmidt number $2$ (see SM~\ref{app:rho33} for a proof), to local dimension $4\times 5$, using one SLOCC extension and two consecutive nontrivial extensions outlined above (See Fig.~\ref{fig:45be} for the schematic representation).  The state $\rho^{3\times 3}$ admits a similar representation as used in Refs.~\cite{krebs2024high,ghimire2023quantum,lockhart2018entanglement}, and is known under the name of \textit{generalized grid} states. We briefly define them here, as the graphical visualization has proven fruitful to exactly calculate their Schmidt number, a task which is notoriously hard for general states.

\textit{High Schmidt number generalized grid states --} 
To depict an unnormalized grid state in $m\times n$, create a $m\times n$ grid $G$.  Identify sites $(i,j)$ of $G$ with a computational basis state $\ket{ij}$. These sites are connected to each other by hyperedges, which are either solid or dashed. The solid ones, $e^+\subseteq G$ corresponds to the superposition $\ket{e^+}:=\sum_{(i,j)\in e^+}\ket{ij}$. The dashed ones always connect two sites and  correspond to the superposition $\ket{e^-}:=\ket{ij}-\ket{kl}$, where $e^-\subseteq G$. 
A graph $(V,E^+,E^-)$, with a set of vertices $V\subseteq G$, the multiset of solid hyperedges $E^+$ and the multiset of dashed edges $E^-$ translates to an unnormalized mixed state, 
\begin{equation}
    \rho = \sum_{e^+\in E^+} \ketbra{e^+}{e^+}+\sum_{e^-\in E^-} \ketbra{e^-}{e^-}.
\end{equation}
One can normalize the state, but this does not change its entanglement properties. 
Generalized grid states give a succinct graphical description of a mixed state, which helps avoid writing out the entire state and  depicts local extensions intuitively. 
We give a $(1,2)$-local extension of $\rho^{3\times 3}$ (see Fig.~\ref{fig:45be} (left)) to local dimension $\rho^{4\times 5}$ (see Fig.~\ref{fig:45be} (right) in three steps:
First, two product states $\ket{30}$ and $\ket{32}$, each with a multiplicity of $3$, are admixed, trivially extending subsystem $A$ via SLOCC transformations.
Then, two nontrivial extensions are performed on subsystem B. The first with $\chi = \ketbra{20}{3}$ does not extend a Schmidt number by Theorem~\ref{th:ext2}, since the local projection $(\mathbb{1}_\mathrm{B}-\ketbra{0}{0}_\mathrm{B})$ results in a separable state (See SM~\ref{app:SingleExtInsuff} for the proof). Finally, we apply the second nontrivial local extension $\chi = \ketbra{02}{3}$, resulting in the state  $\rho^{4\times 5}$. We next prove $\mathrm{SN}(\rho^{4\times 5})=3$, which is the lowest dimensional Schmidt number known in literature for given dimensions~\cite{pal2019class,krebs2024high}. 

\begin{proof}
We analyze $\rho^{4\times 5}$ using the range criterion~\cite{sanpera2001schmidt}. 
We prove that every Schmidt rank $2$ state in $R(\rho^{4\times 5})$ is orthogonal to $\ket{e^*}:=\ket{00}+\ket{11}+\ket{22}$. This is sufficient to conclude $\mathrm{SN}(\rho)>2$, as $\ket{e^*}\in R(\rho^{4\times 5})$.
We parametrize  a state in the range $\ket{\psi}\in R(\rho)$ as a coordinate matrix  $\mathbf{\Psi}_{ij}= \mean{ij|\psi}$, 
\begin{equation}
  \mathbf{\Psi}: = 
    \begin{pmatrix}
        \psi_{00} & \psi_{01} & \psi_{02} &     0     & \psi_{04}\\
        \psi_{10} & \psi_{00} & \psi_{01} &     0     &     0\\
        \psi_{20} &-\psi_{10} & \psi_{00} & \psi_{23} &     0\\
        \psi_{30} &     0     & \psi_{32} & \psi_{20} & \psi_{02}
    \end{pmatrix}.
\end{equation}
Naturally, the rank $\mathrm{rk}(\mathbf{\Psi})$ equals $\mathrm{SR}(\ket{\psi})$, i.e. for $\mathrm{SR}(\ket{\psi})\le 2$, all sets containing $3$ columns of $\mathbf{\Psi}$ must be linearly dependent. 
This is the case if all $3\times 3$ minors of  $\mathbf{\Psi}$ vanish, translating into the sets of constraints, a subset of which we present here, $g_i$:
\begin{align}
\begin{split}
    g:=\{&(\psi_{20}(\psi_{00}^2-\psi_{01}\psi_{10}),\psi_{02}(\psi_{00}^2+\psi_{01}\psi_{10}),\\
    &\psi_{20}(\psi_{01}^2-\psi_{00}\psi_{02}),-\psi_{02}(\psi_{10}^2+\psi_{00}\psi_{20}),\\
    &\psi_{00}^3+\psi_{01}^2\psi_{20}-\psi_{10}^2\psi_{02}-\psi_{00}\psi_{02}\psi_{20})\}=0.
\end{split}
\end{align}
By Hilbert’s Nullstellensatz, to prove that every solution of the system satisfies 
$\psi_{00} = 0$, it suffices to show that some power of $\psi_{00}$ lies in the ideal 
generated by the $g_i$. In other words, we need an identity of the form $ \sum_i p_i g_i = \psi_{00}^k$ for suitable polynomials $p_i$ and $k\in \mathbb{N}$. In fact, this can already be achieved with very simple  
choices of $p_i$: a direct calculation shows that $\psi_{00}(g_5 - g_3 - g_4) - \psi_{02} g_1 + \psi_{20} g_2 \;=\; \psi_{00}^4.$ Thus, $\psi_{00}^4$ belongs to the ideal generated by $\{g_1,\dots,g_5\}$, and the 
Nullstellensatz guarantees that every common zero of the $g_i$ must satisfy 
$\psi_{00} = 0$, concluding the proof.
\end{proof}

\balance{
To conclude, we discuss the scaling of the Schmidt number with increasing local dimensions. 
In earlier work, an infinite family of Schmidt number $k$ PPT states was constructed in 
nonhomogeneous $(2k-1) \times \tfrac{k(k+1)}{2}$ systems~\cite{krebs2024high}. 
Here, we improve upon those results by proposing a family of generalized grid states in 
homogeneous $(2k-1) \times (2k-1)$ systems with the same Schmidt number, thereby reducing 
the second dimension to linear scaling. Explicit examples with Schmidt number 
$k \in \{2,3,4,5\}$ are provided in SM~\ref{app:scaling}. In each case, the construction 
embeds lower Schmidt number states as local blocks, so that the state with Schmidt number $k+1$ arises 
as a $(2,2)$-extension of the one with Schmidt number $k$. 
We conjecture that this recursive structure extends to larger dimensions. Rather than 
providing a complete proof, we highlight that Theorem~\ref{th:4x5ex} suggests the existence 
of even more favorable candidates once local extensions are generalized beyond the simple 
forms considered here. This opens a promising direction for future work, pointing toward 
increasingly efficient families of high-dimensional PPT states.

\textit{Discussions --} The characterization of entanglement in mixed states remains one of the most difficult 
open problems in quantum information theory, with few general tools available beyond low 
dimensions.
In this work we have introduced the method of local extensions, the natural 
dual to peel-off projections, as a constructive principle for building families of PPT 
states with controlled growth of the Schmidt number. This framework not only yields 
general bounds on how entanglement dimensionality can change under projections and 
extensions, but also enables explicit constructions that push the known limits of 
high-dimensional bound entanglement.  

Our results advance the broader effort to chart the landscape of high-dimensional 
mixed-state entanglement, an area where progress has been very limited. 
By providing a systematic way to construct and certify high-Schmidt-number PPT states, 
local extensions enrich the sparse set of techniques available for tackling the structure 
of bound entanglement. While we do not attempt a general proof of scaling behavior, our 
analysis indicates that even more efficient families are likely to emerge once extensions 
are generalized beyond the forms introduced here. Such a proof would likely be of a more 
technical nature, and we leave this line of work to future studies that can further 
develop and refine the framework. Beyond our immediate impact, the local extension 
framework may also inform studies of optimal distillation protocols and the 
incompressibility of quantum resources, thereby deepening the understanding of 
entanglement as a quantum information resource.

\textbf{Acknowledgments} -- We thank Otfried Gühne, Nikolai Miklin, Lucas Vieria, Albert Rico, Sophia Denker and Zaw Lin Htoo for discussions.
}

\onecolumngrid

\appendix

\begin{center}
\vfill \textbf{Overview of the Supplemental Material}
\end{center}
The first SM~\ref{app:SCpeel} contains Theorem~\ref{th:ext2} as a statement about $k$-positive maps, alongside a proof elucidating its nature as a Schur complement identity.

The two SM~\ref{app:ext_psd} and~\ref{app:ext_ppt} characterize the extremal geometry of the respective convex cones of PSD and PPT local extensions, which found no immediate use in the main text, but is a useful tool when optimizing over local extensions or making other global statements.

SM~\ref{app:rho33} discusses several peculiar geometric properties exhibited by $\rho^{3\times 3}$, the starting point of the construction in Theorem~\ref{th:4x5ex}. Afterward, in SM~\ref{app:SingleExtInsuff}, we prove that the intermediate steps before obtaining $\rho^{4\times5}$ remain SN $2$ states.
Finally, in SM~\ref{app:scaling}, we introduce new states $\rho^{(k)}$, which for $k\le 5$ have Schmidt number $k$ in dimensions $(2k-1)\times(2k -1)$.

\section{Proof for the Schur-complement peel-off theorem}
\label{app:SCpeel}
We translate our strengthened projection property in Theorem~\ref{th:ext2}  to a statement about $k$-positive maps, resulting in a Schur complement identity.
\begin{theorem}[Schur Complements for $k$-positivity]
Let $W_k\in\bounded{m}{n}$ be the Choi matrix of a $k$-positive map, with $k>1$, that is, a Schmidt number $\le k+1$ witness. 
Now, split $W_k$ by singling out a $1\times n$ subspace, i.e., find some operator $W_c \fsub W_k$. The associated Schur complement $W_{c\backslash e}$ is then the Choi matrix of a $(k-1)$-positive map.
\end{theorem}
\begin{proof}
    Write out a pure state on $\ket{\psi} \in \H_\mathrm{A}\otimes\H_\mathrm{B}$ in the split form: $\ket{\psi} = \begin{pmatrix}
            \ket{\psi^c}\\
            \ket{\psi^e}
    \end{pmatrix}$. Then, 
    \begin{align}
    \begin{split}
        \bra{\psi}W_k\ket{\psi} = \bra{\psi^c}W_c\ket{\psi^c} +  \bra{\psi^e}W_e\ket{\psi^e} + \bra{\psi^c}\chi\ket{\psi^e}+\bra{\psi^e}\chi^\dagger\ket{\psi^c} = \\
        \bra{\psi^c}(W_c-\chi W_e^{-1} \chi^\dagger) \ket{\psi^c}+
        (\bra{\psi^e}+\bra{\psi^c}\chi W_e^{-1} ) W_e
        (\ket{\psi^e}+W_e^{-1}\chi^\dagger\ket{\psi^c} )    .
        \end{split}
    \label{eq:SquareSplit}
    \end{align}
    This manipulation is only valid if $R(\chi^\dagger)\subseteq R(W_e)$, indeed, this is always the case for any $k$-positive map with $k>1$, as otherwise, we could construct a Schmidt rank $2$ state with a negative expectation value.
    To see this, suppose the contraposition; that there exists a $\ket{\bot\gamma}$, such that $W_k\ket{\bot\gamma} = \chi\ket{\gamma}\neq 0$. Then, there exists an $\ket{\alpha\beta}\in \mathcal{H}\lA\otimes\mathcal{H}\lB$ with $\bra{\alpha\beta}\chi\ket{\gamma}\neq 0$, which implies that the superposition $\ket{\phi_z}=\ket{\bot\gamma}+z\ket{\alpha\beta}$ has $\bra{\phi_z}W\ket{\phi_z}<0$ for some $z$. 

    Minimizing the expression from Eq.~\ref{eq:SquareSplit} in $\ket{\psi^e}$ leaves $\bra{\psi^c}(W_c-\chi W_e^{-1} \chi^\dagger \ket{\psi^c}= \bra{\psi^c}W_{c\backslash e} \ket{\psi^c}$. We can deduce that $W_{c\backslash e}$ has a $(k-1)$-positive Choi map. Otherwise, a $\ket{\psi}$ with Schmidt rank $k$ would attain a negative expectation value. This allows to decompose $W$ into a Schmidt number $\le k$ witness $W_{c\backslash e}$ and a PSD operator by writing: 
    \begin{equation}
        W = \begin{pmatrix}
            W_{c\backslash e} & 0 \\
            0 & 0
        \end{pmatrix}+ \begin{pmatrix}
            \chi W_e^{-1} \chi^\dagger & \chi\\
            \chi^\dagger & W_e
        \end{pmatrix}.
    \end{equation}
    Here, we learned that in a Peel-Off theorem like in Ref.~\cite{yang2016all}, the PSD part can always be chosen as a low-rank flat extension of a $1\times n$-dimensional $W_e$, matching the subtracted component in the Schur-complement.
\end{proof}
%\robin{checked times 1}

\section{Convex geometry of PSD local extensions}
\label{app:ext_psd}
The set of PSD local extensions of a fixed (normalized) state $\rho_c$ consists of all local extensions 
$\{\rho\fsup \mu\rho_c|\rho\ge 0,~\mu\in \mathbbm{R}_+\}$
This set is closed under arbitrary conic combinations. In general, local extensions are usually constrained to preserve some property of an initial state, like positive semidefiniteness or the PPT property. 
We note that extreme points of the intersection between the set of arbitrary hermitian local extensions of $\rho_c$ and another convex cone $\mathcal{C}$ will generally not be extreme points of $\mathcal{C}$.

We denote a local extension by the triple $(\mu,\chi,\rho_{e\backslash c})$, where $\mu=\tr{\rho_c}$, as \begin{equation}
    \rho_{(\mu,\chi,\rho_{e\backslash c})} \equiv \begin{pmatrix}
       \rho_c & \chi \\
        \chi^\dagger &  \chi^\dagger\rho_c^{-1} \chi + \rho_{e\backslash c}
    \end{pmatrix}.
\end{equation}
A local extension $\rho_{(\mu,\chi,0)}$ is called \textit{flat},  i.e. $\rho_{e\backslash c} = 0$. 

\begin{theorem}
Every flat local extension $\rho$ is extremal and every (PSD) local extension can be expressed as the convex sum of a flat local extension and product states in the range of $\rho_{e\backslash c}$.
\end{theorem}

\begin{proof} 
    We first introduce the notation  $\mathbb{1}(S)$ for the unique self-adjoint projector on a linear subspace $S = \mathbb{1}(S)(\H_\mathrm{A}\otimes \H_\mathrm{B}) \subseteq \H_\mathrm{A}\otimes \H_\mathrm{B}$.
    For general PSD matrices $A$, $B$ and sufficiently small $\epsilon\in \mathbb{R}_+$,  $A-\epsilon B\ge 0$ iff $R(B)\subseteq R(A)$. For a flat local extension defined by $(\mu,\chi,0)$, it holds that
    \begin{equation}
        R(\rho) = R(\begin{pmatrix}
               \mathbb{1}(R(\rho_c))\\
               \chi^\dagger\rho_c^{-1}
                \end{pmatrix}),
        \end{equation}
        where we can immediately infer that every state $\ket{\psi^c}\in R(\rho_c)$ is assigned an  extension $\ket{\psi^c}+\chi^\dagger\rho_c^{-1}\ket{\psi^c}$, which uniquely completes it to an element of $R(\rho)$.  This implies directly that $R(\rho_{(\mu,\chi,0)})$ contains $R(\rho_{(\gamma,\xi,0)})$ iff $\chi^\dagger\rho_c^{-1} = \xi^\dagger\rho_c^{-1}$.  Hence, any flat local extension ceases to be PSD when subtracting any other (flat) local extension.  This defines the set of extreme points of the set of PSD local extensions: Flat local extensions of the form $\rho_{(\mu,\chi,0)}$ and rank $1$ projectors, which formally also are local extensions $\rho_{(0,0,\ket{\phi}\bra{\phi})}$.
\end{proof}
This concludes our discussion of the cone of PSD local extensions, which allows us to intersect with the PPT cone, yielding the cone of PPT local extensions.  
\section{Convex geometry of PPT local extensions}
\label{app:ext_ppt}
Next, we give an algebraic characterization of the extremal points of the cone of PPT local extensions, helping the discovery of new high Schmidt number PPT states.
If subsystem $\mathcal{H}_\mathrm{A}$ is locally extended, provided we define, without loss of generality, $\ket{\bot} = \ket{\bot}^*$, denote by $\rho^{\Ta}_{e\backslash c}\equiv \rho_e -\chi^{\Ta}\rho_c ^{\Ta} (\chi^{\dagger})^{\Ta}$ and the partial transposition takes the form 
\begin{equation}
 \rho^{\Ta} = 
\begin{pmatrix}
        \rho_c ^{\Ta}& (\chi^{\dagger})^{\Ta}\\
        \chi^{\Ta} & \rho_e        
\end{pmatrix}.
\end{equation}

\begin{theorem} A PPT local extension $\rho$ is extremal if the following intersection $R(\rho_{c})_{\mathrm{AB}}\otimes R(\rho_{e\backslash c})_{\mathrm{C}}\cap R(\rho^{\Ta}_{c})_{\mathrm{AC}}\otimes R(\rho^{\Ta}_{e\backslash c})_\mathrm{B}$ is one-dimensional and 
$R(\rho_{e\backslash c})\cap R(\left(\rho^{\Ta}\right)_{e\backslash c}) = \{0\}$.
\end{theorem}

\begin{proof}
Using Schur's complement, $\rho_e \ge \chi^{\dagger}\rho_c  \chi$ and $\rho_e \ge\chi^{\Ta}\rho_c ^{\Ta} (\chi^{\dagger})^{\Ta}$. Hence, given a fixed $\chi$, there is not a unique $\rho_e$ that defines a minimal local extension like in the PSD case, since the Loewner order of PSD matrices is not a total order~\cite{kadison1951order}. We can, however, directly conclude that for every \textit{extremal PPT extension }
\begin{equation}
R\left(\rho_{e\backslash c}\right)\cap R\left(\left(\rho^{\Ta}\right)_{e\backslash c}\right) = \{0\}.
\label{eq:trivint}
\end{equation}
{
This is not sufficient; the qubit state $\ket{\phi_+}\bra{\phi_+}+\frac{1}{2}\ketbra{01}{01}+\frac{1}{2}\ketbra{10}{10}$ is a counterexample, as it is a local extension of $1\otimes\mathbbm{1}/2$ and satisfies Eq.~\ref{eq:trivint}, but decomposes as $\frac{1}{2}(\ketbra{\phi_+}{\phi_+}+\ketbra{\psi_+}{\psi_+})+ \frac{1}{2}(\ketbra{\phi_+}{\phi_+}+\ketbra{\psi_-}{\psi_-})$.
}
To derive the full criterion, observe that $R(\rho_{(\mu',\xi,\sigma_{e\backslash c})})\subseteq R(\rho_{(\mu,\chi,\rho_{e\backslash c})})$ when the following invertible matrix $M$ exists:
\begin{equation}
    \exists M' ~\mathrm{s.t.} ~ \rho_{(\mu',\xi,\sigma_{e\backslash c})}M'= \rho_{(\mu,\chi,\rho_{e\backslash c})} \iff  \exists M ~\mathrm{s.t.} ~      \begin{pmatrix}
        \mathbb{1}_c & 0 \\
        \xi^\dagger\rho_c^{-1} & \sigma_{e\backslash c}
    \end{pmatrix} M =
    \begin{pmatrix}
        \mathbb{1}_c & 0 \\
        \chi^\dagger\rho_c^{-1} & \rho_{e\backslash c}
    \end{pmatrix}.
\end{equation}
To obtain the second identity, consider the invariance $R(X) = R(XY)$ for an arbitrary operator $X$ and an invertible $Y$. Without losing generality, we can further choose:  
\begin{equation}
     M = \begin{pmatrix}
        \mathbb{1}_c & 0 \\
        C & D
    \end{pmatrix} \implies 
    \begin{pmatrix}
        \mathbb{1}_c & 0 \\
        \xi^\dagger\rho_c^{-1}+\sigma_{e\backslash c}C & \sigma_{e\backslash c} D
    \end{pmatrix} = 
    \begin{pmatrix}
        \mathbb{1}_c & 0 \\
        \chi^\dagger\rho_c^{-1} & \rho_{e\backslash c}
    \end{pmatrix}
\end{equation}
and hence conclude from $(\chi^\dagger-\xi^\dagger)\rho_c^{-1} = \sigma_{e\backslash c}C = \rho_{e\backslash c}DC$ that
\begin{equation}
    \Delta^\dagger \equiv \chi^\dagger-\xi^\dagger \in  \mathcal{B}\left(R\left(\rho_{c}\right),R\left(\rho_{e\backslash c}\right)\right),
\end{equation}
giving a new condition. The same holds for the partial transpose, yielding the further identity.
\begin{equation}
    (\Delta)^{\Ta} = \chi^{\Ta}-\xi^{\Ta} \in  \mathcal{B}\left(R\left(\rho^{\Ta}_{c}\right),R\left(\rho^{\Ta}_{e\backslash c}\right)\right),
\end{equation}
We can eliminate the partial transposes like in Eq.~\eqref{eq:linext},  showing that a PPT local extension is extremal, if and only if the 
vector spaces $R(\rho_{c})_{\mathrm{AB}}\otimes R(\rho_{e\backslash c})_{\mathrm{C}}$ and $R(\rho^{\Ta}_{c})_{\mathrm{AC}}\otimes R(\rho^{\Ta}_{e\backslash c})_\mathrm{B}$ have a nonzero one-dimensional
intersection and Eq.~\eqref{eq:trivint} holds.
\end{proof}

\section{Characterization of $\rho^{3\times 3}$}
\label{app:rho33}
Previous entanglement criteria for grid states used product states in the kernel of the grid state, usually corresponding to sites not incident to any edges in the corresponding graph. 
Given enough product states in the kernel, some minors used in defining the variety of SR $k$ states in $R(\rho)$ turn into monomials, simplifying considerably.
These minors are then very useful in proofs based on direct algebraic elimination~\cite{lockhart2018entanglement,ghimire2023quantum,krebs2024high}. 
We define $\rho^{3\times 3}$, an example of a grid state without a product state in the kernel, yet with entanglement amenable to the range criterion by the following pure states
\begin{equation}
    \ket{e_0} = {\ket{00} +\ket{11}+\ket{22}},\  \ 
     \ket{e_1} = {\ket{01}+\ket{12}},\  \  \ket{e_2} ={\ket{10}-\ket{21}},\ \ 
       \ket{e_3} ={\ket{02}},\ \ \ket{e_4} = 
    {\ket{20}},
\end{equation}
and weights $r_i = (1,1,1,3,3)$ as $\rho^{3\times 3} = \sum_i r_i \ket{e_i}\bra{e_i}$, equivalently given in Fig.~\ref{fig:45be}.
We determine the 
range coordinate matrix as: 
\begin{equation}
    \mathbf{\Psi} = \{\mean{ij|\psi }\}_{ij} = 
    \begin{pmatrix}
        \psi_{00} &  \psi_{01} & \psi_{02}\\
        -\psi_{10} &  \psi_{00} & \psi_{01}\\
        \psi_{20} &  \psi_{10} & \psi_{00}\\
    \end{pmatrix},
\end{equation}
then we consider two of the minor equations satisfied by all product states
\begin{equation}
    \left|\begin{matrix}
        \psi_{00} &  \psi_{01}\\
        -\psi_{10} &  \psi_{00} 
    \end{matrix} \right| = \psi_{00}^2+\psi_{01}\psi_{10}=0,~
    \left|\begin{matrix}
        \psi_{00} &  \psi_{01}\\
        \psi_{10} & \ \  \psi_{00} 
    \end{matrix} \right| = \psi_{00}^2-\psi_{01}\psi_{10}=0,
\end{equation}
which together imply that $\psi_{00}^2=0$ and $\psi_{01}\psi_{10}=0$. 
The identity $\psi_{00}=0$ already establishes that a component of $\rho^{3\times 3}$ cannot be spanned using product states, hence $\rho^{3\times 3}$ is entangled. Plugging $\psi_{00}=0$ in the coordinate matrix gives two further equations, $\psi_{01}^2=0$ and $\psi_{10}^2=0$. Hence, $\ket{e_3}$ and $\ket{e_4}$ are the only two product states in the range of $\rho^{3\times 3}$, and neither of them has a partial transpose in the range of $(\rho^{3\times 3})^{\Tb}$, which makes $\rho^{3\times 3}$ an edge state~\cite{lewenstein2000optimization}. While evaluating the range criterion is straightforward for this state, it has a higher formal complexity than the algebraic approaches in~\cite{lockhart2018entanglement,ghimire2023quantum,krebs2024high}, since two minors need to be combined to give the required monomial identities $\psi_{00}^2=0$ and $\psi_{01}\psi_{10}=0$. Furthermore, the partial transpose has a similar simple form:
\begin{equation}
    \ket{f_0} =  
    {\ket{02}+\ket{11}-\ket{20}},~~
\ket{f_1}=    {\ket{02}+\ket{20}},~~
\ket{f_2}=      {\ket{01}+\ket{10}},~~ \ket{f_3}=  {\ket{12}+\ket{21}},~~
\ket{f_4}=     {\ket{00}},~~
\ket{f_5}=      {\ket{22}},
\end{equation}
with weights $\Tilde{r}_i = (1,2,1,1,1,1)$ as
\begin{equation}
    (\rho^{3\times 3})^{\Tb} = \sum_i \Tilde{r}_i\ket{f_i}\bra{f_i}, 
\end{equation}
hence the rank changes under partial transposition, in this case from $5$ to $6$. 
We also observe that when setting $\psi_{02}=\psi_{20}=0$, the $3\times 3$ minor $\det \mathbf{\Psi}$ becomes
\begin{equation}
    \left|\begin{matrix}
        \psi_{00} &  \psi_{01} & 0\\
        -\psi_{10} &  \psi_{00} & \psi_{01}\\
        0 &  \psi_{10} & \psi_{00}\\
    \end{matrix}\right| = \psi_{00}^3.
\end{equation}
This property hints at the existence of Schmidt number $3$ extensions of $\rho^{3\times 3}$, which are confirmed in Theorem~\ref{th:4x5ex}.

\section{A single extension of $\rho^{3\times 3}$ is insufficient}
As discussed in the main text, after performing a trivial extension on subsystem A of $\rho^{3\times 3}$, and a single nontrivial extension with $\ketbra{20}{3}$ on subsystem B, the resulting state $\rho^{4\times 4}$ still has SN $2$. We shall prove this by finding a suitable separable projection on the $\rho^{4\times 4}$ and applying Theorem~\ref{th:ext2}.
\label{app:SingleExtInsuff}
\begin{figure}[H]
    \centering
    \includegraphics[width=0.5\linewidth]{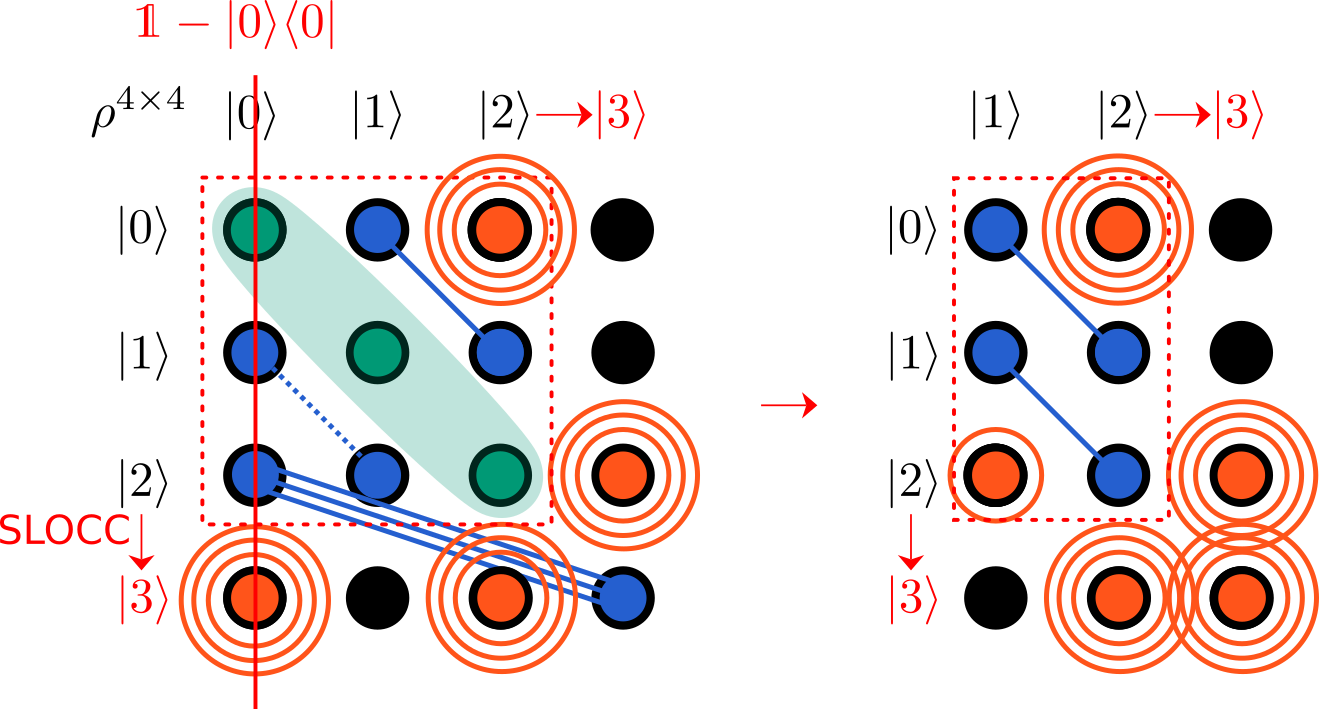}     
    \caption{l.: Grid state obtained by performing a first extension, later projection indicated by red line. r.: Grid state after applying the local projection. Decomposes into one $2\times 3$ (indicated by the dashed red box) and several $1\times 1$ separable local blocks.}
    \label{fig:sep_proj}
\end{figure}
After performing the single local extension with $\ketbra{20}{3}$ and projecting with $\mathbbm{1}\lB-\ket{0}\bra{0}\lB$, the remaining state is a convex mixture of product states $\ket{30}$, $\ket{32}$, $\ket{23}$, $\ket{33}$  and a $2\times 3$ PPT state supported on the local block $\mathrm{span}(\ket{0},\ket{1},\ket{2}) \otimes \mathrm{span}(\ket{1},\ket{2})$, as indicated in Fig.~\ref{fig:sep_proj}. The latter, according to the Peres-Horodecki criterion~\cite{horodecki1997separability} is also separable, showing overall separability of the projected state. According to Theorem~\ref{th:ext2}, the existence of this separable projection implies that $\rho^{4\times 4}$ only has the Schmidt number $2$.

\section{Schmidt number scaling}
\label{app:scaling}
Here, we give $\rho^{(c)}$ explicitly in terms of its spectrum and eigenvectors, which coincide with edges in the grid state formalism, as can be seen in Fig.~\ref{fig:scaling} on the example of $\rho^{(4)}$.
The states discussed here have a single large hyperedge supported on $k$ sites, and otherwise only edges supported on $\le 2$ sites. 
\begin{figure}[H]
    \centering
    \includegraphics[width=0.8\linewidth]{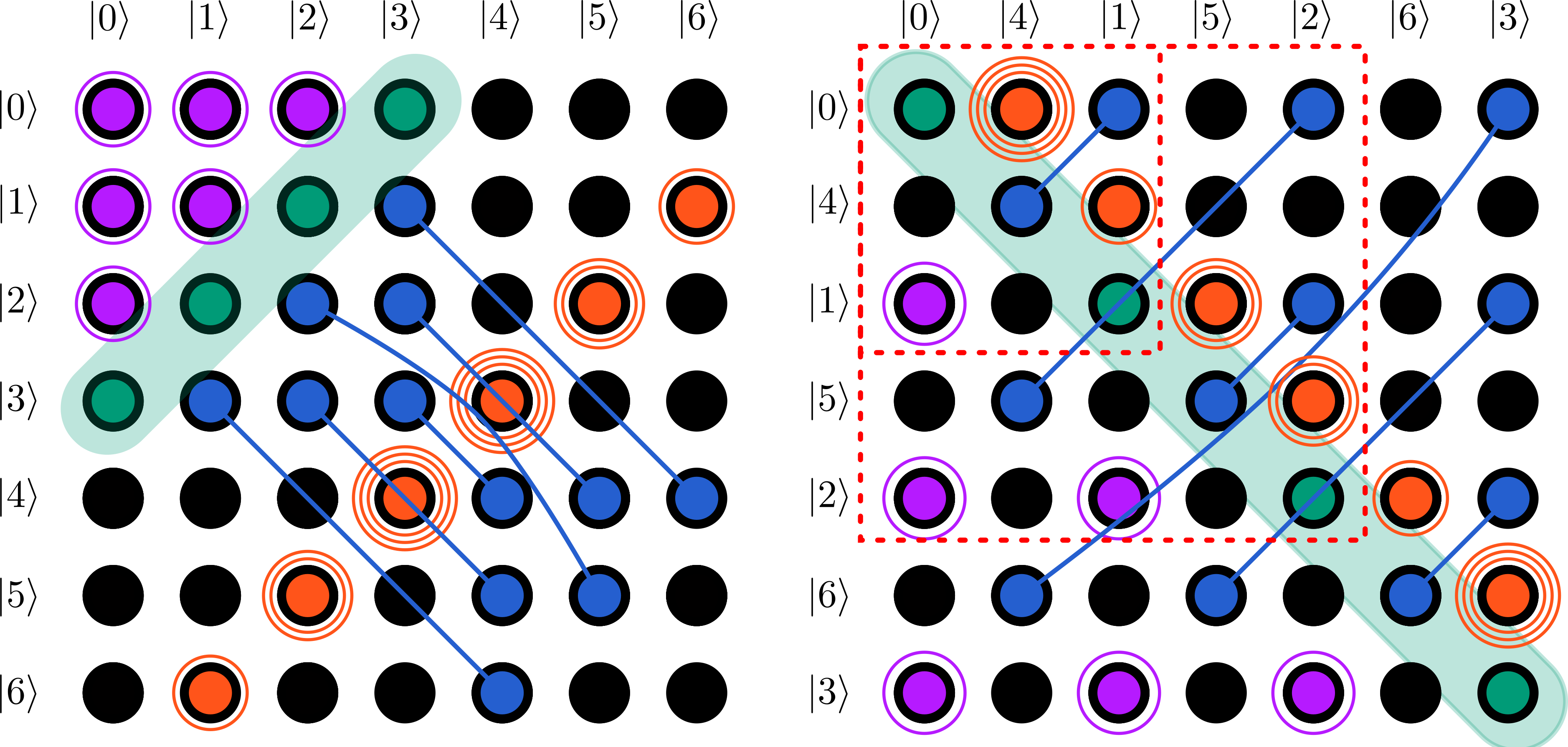}
\caption{Graph representing $\rho^{(4)}$, with $\ket{\alpha}$, $\ket{\beta_{ij}}$, $\ket{\gamma_{ij}}$, $\ket{\delta_i}$ indicated as colors green, blue, purple, orange, respectively. The representation on the left side prioritizes contiguous empty sites, while the representation on the right-hand side emphasizes the extension structure.}
    \label{fig:scaling}
\end{figure}

We explicitly write $\rho^{(k)}\in\bounded{2k-1}{2k-1}$ for any $k$ in terms of its edges, which are also eigenvectors:
\begin{align}
    &\,\ket{\alpha} = \sum_{i=0}^{k-1} \ket{i,k-1-i}, &\ket{\beta_{ij}}&=\sum \ket{i,j} + \ket{2k-1-j,2k-1-i}~(i+j\ge k,~i<k,~j<k),\nonumber\\
    &\,\ket{\gamma_{ij}} = \ket{i,j}~(i+j<k-1), &\ket{\delta_i}\hphantom{_j}&=\ket{i,2k-1-i}~(1\le i\le 2k-2).
\end{align}
\begin{equation}
    \rho^{(k)} = \ketbra{\alpha}{\alpha} + \sum_{ij} \ketbra{\beta_{ij}}{\beta_{ij}} + \sum_{ij} \ketbra{\gamma_{ij}}{\gamma_{ij}} + \sum d_i \ketbra{\delta_i}{\delta_i}.
\end{equation}
We proceed to show that there exist $d_i\ge 0$,  such that $\rho^{(k)}$ is PPT. To this end, we give an explicit decomposition of $\rho^{(k),T_\mathrm{A}}$ into Schmidt rank $\le 2$ quantum states.

To achieve this, we utilize the sparsity of $\rho^{(k)}$. We only need to consider off-diagonal elements from $\ketbra{\alpha}{\alpha}$ and from $\ketbra{\beta_{ij}}{\beta_{ij}}$. 
The first set of off-diagonal elements  has partial transposes $\ketbra{i,j}{k-1-j,k-1-i}$ for $i\le k$, $i\le j$, which taking into account the diagonal elements from $\ket{\beta_{ij}}$ and $\ket{\gamma_{ij}}$ combine into size $2$ blocks $\ketbra{\eta_{ij}}{\eta_{ij}}$ with $\ket{\eta_{ij}} := \ket{i,j}+\ket{k-1-j,k-1-i}$. 
The second set of off-diagonal elements has partial transposes $\ketbra{i,2k-1-i}{j,2k-1-j}$, for $i+j\ge k $, $i<k$, $j<k$,  which connect together with the $d_i$ into another block $D$ of size $2k$. 
By assigning sufficiently high values to the $d_i$, we guarantee that $D$ is PSD. In the present case, the requisite minimal values are $d_i \equiv \min(i,2k-1-i)$, as then $D$ has a kernel element $\ket{\Omega}\equiv \sum_{i=1}^{k-1} (\ket{i,2k-1-i}-\ket{2k-1-i,i})$, with $\mean{\Omega|\delta_i}\neq 0$. Hence,  $D-\epsilon\ketbra{\delta_i}{\delta_i}\not\ge 0$ for all $\epsilon>0$, thus the $d_i$ are minimal~\cite{ghimire2023quantum}. 
To prove $\mathrm{SN}(\rho^{(k),T_\mathrm{A}})=2$, we can further define $\ket{\mu_{ij}}:= \ket{i,2k-1-i} + \ket{j,2k-1-j}$, and observe that $D = \sum_{i+j\ge k,~i<k,~j<k} \ketbra{{\mu_{ij}}}{{\mu_{ij}}}$, with the remainder of $\rho^{(k),T_\mathrm{A}}$ decomposing into the $\ket{\eta_{ij}}$-components and further product states appearing as a diagonal block.

We determine the $\mathrm{SN}(\rho^{(k)})$ for $k\le 5$ with the range criterion. Note that this result is completely independent from the value of the $d_i$, since the range does not depend on it. In general, this implies that for any given value of $d_i$, the highest attainable $\rho^{(k)}$ is bounded, in particular, choosing the minimal values of the $d_i$ prohibits further extension in the above manner.
We parametrize $R(\rho^{(k)})$ with the coordinates 
\begin{equation}
     \ket{\psi(\alpha,\beta,\gamma,\delta)}:={\alpha}\ket{\alpha} + \sum_{ij} {\beta_{ij}}\ket{\beta_{ij}} + \sum_{ij} {\gamma_{ij}}\ket{\gamma_{ij}} + \sum_i  {\delta_i}\ket{\delta_i} \in R(\rho^{(k)})
\end{equation}
and like in the proof of Theorem~\ref{th:4x5ex}, we consider the ideal generated by $k\times k$ minors, which has the Schmidt rank $(k-1)$ states in this space as its variety. Since all minors considered here have integer coefficients, we can choose the ring of polynomials with rational coefficients as the base ring. This choice allows exact computer-assisted proofs of ideal membership, avoiding the issues accompanying floating-point arithmetic.  To this end, we confirm ideal membership of $\alpha^k$ for $k\le 5$ with the computer-algebra system SAGE~\cite{sagemath}, which computes a Groebner basis for the system of minors and proceeds to reduce $\alpha^k$ by it. 
Similar to the example in Theorem~\ref{th:4x5ex}, considering a subset of minors is already sufficient: For all examples tested, we can disregard any minor containing one of the $\delta_i$.
With the Nullstellensatz, this proves that all Schmidt rank $(k-1)$ states in $R(\rho^{(k)})$ need to obey $\alpha=0$, hence are orthogonal to $\ket{\alpha}\in R(\rho^{(k)})$. Thus, the $\rho^{(k)}$ have Schmidt number $k$ for $k\le 5$, providing the lowest dimensional known examples of Schmidt number $4$ and $5$, improving over the previous examples in~\cite{pal2019class,krebs2024high}, which had dimensions $2k\times 2k$ and $(2k-1) \times \frac{k(k+1)}{2}$ respectively.

\twocolumngrid
%\bibliography{bibl.bib}

\end{document}